\documentclass[12pt]{amsart}
\usepackage{amssymb,latexsym,amsmath,amscd,amsthm}
\usepackage[latin1]{inputenc}
\usepackage{epsfig}
\usepackage[active]{srcltx}
\usepackage{psfrag}
\usepackage{enumerate}
\newcommand{\GMM}{\mathtt{GMM}}
\newcommand{\SSM}{\mathtt{SSM}}
\newcommand{\Ka}{\mathtt{K81}^{\ast}}
\newcommand{\Kb}{\mathtt{K80}^{\ast}}
\newcommand{\JC}{\mathtt{JC69}^{\ast}}

\newcommand{\RR}{\mathbb{R}}
\newcommand{\CC}{\mathbb{C}}

\newcommand{\lra}{\longrightarrow}

\newcommand{\al}{\alpha}

\newcommand{\tr}{\operatorname{tr}}

\renewcommand{\a}{\mathtt{A}}
\renewcommand{\c}{\mathtt{C}}
\newcommand{\g}{\mathtt{G}}
\renewcommand{\t}{\mathtt{T}}

\theoremstyle{definition}
\newtheorem{defi}{Definition}[section]

\theoremstyle{plain}
\newtheorem{lema}[defi]{Lemma}
\newtheorem{thm}[defi]{Theorem}

\newtheorem{prop}[defi]{Proposition}

\newtheorem{rk}[defi]{Remark}

\newtheorem{alg}[defi]{Algorithm}

\newtheorem{teo-def}[defi]{Theorem/Definition}

\newenvironment{proofPropK81}{\noindent {\textit{Proof of Proposition \ref{PropK81}.}}}{\quad \hfill $\Box$}

\title{Generating Markov evolutionary matrices for a given branch length}
\author{Marta Casanellas}
\address{Departament de Matemàtica Aplicada I. ETSEIB. Universitat Polit\`ecnica de Catalunya. Avinguda Diagonal 647. 08028 Barcelona. Spain.}
\email{marta.casanellas@upc.edu}
\author{Anna Kedzierska}
\address{Departament de Matemàtica Aplicada I. ETSEIB. Universitat Polit\`ecnica de Catalunya. Avinguda Diagonal 647. 08028 Barcelona. Spain.}
\email{anna.kedzierska@upc.edu}
\thanks{Both authors are partially supported by Generalitat de Catalunya, 2009 SGR 1284. Research of the first author partially supported by Ministerio de Educaci\'on y Ciencia MTM2009-14163-C02-02.}
\begin{document}

\begin{abstract}
Under a markovian evolutionary process, the expected number of substitutions per site (also called \textit{branch length}) that have occurred when a sequence has evolved from another according to a transition matrix $P$ can be approximated by $-\frac{1}{4}\log \det P.$  When the Markov process is assumed to be continuous in time, i.e. $P=\exp Qt $ it is easy to simulate this evolutionary process for a given branch length (this amounts to requiring $Q$ of a certain trace). For the more general case (what we call \textit{discrete-time models}), it is not trivial to generate a substitution matrix $P$ of given determinant (i.e. corresponding to a process of given branch length). In this paper we solve this problem for the most well-known discrete-time models $\JC$, $\Ka$, $\Kb$, $\SSM$ and $\GMM$.  These models lie in the class of nonhomogeneous evolutionary models. For any of these models we provide concise algorithms to generate matrices $P$ of given determinant. Moreover, in the first four models, our results prove that any of these matrices can be generated in this way. Our techniques are mainly based on algebraic tools.

\end{abstract}
\maketitle

\section{Introduction}
Phylogenetic reconstruction methods are usually tested on
simulated data, i.e. DNA (or protein) sequences that have been
randomly generated following a molecular evolutionary model on a
phylogenetic tree. It is easy to generate a random DNA sequence
that evolves from a given DNA sequence under a given evolutionary
model if no more constrains are required: one just needs to give
random values to the parameters of the model and generate data
according to the conditional probabilities obtained from the
parameters. An extra effort is needed if
the amount of ``substitution events" is fixed; this magnitude is
usually called the \textit{branch length} of the edge relating
both sequences in the phylogenetic tree.

We will assume (as it is commonly done) that
sites in a DNA sequence are independent and identically distributed (\textit{iid} hypothesis), so that one just models the evolution of
one site (thought as a random variable taking values in
$\{\a,\c,\g,\t\}$). The most common molecular evolutionary models used
in phylogenetics are the so-called \textit{continuous-time
models}. In these models, the substitution events along an edge
$e$ of a rooted phylogenetic tree occur following a continuous-time Markov process:
there is an instantaneous mutation rate matrix $Q$ (usually fixed
throughout the tree) that operates at intensity $\lambda_e$ and
for duration $t_e$ so that the substitution matrix (or transition
matrix) $P_e$ equals $ \exp(Q\cdot \lambda_e t_e).$ Among them
there are the time-reversible models Jukes-Cantor \texttt{JC69}
\cite{JC69}, Kimura two-parameters \texttt{K80} \cite{Kimura1980},
Kimura three-parameters \texttt{K81} \cite{Kimura1981},
\texttt{HKY} \cite{Hasegawa1985}, and \texttt{GTR} \cite{GTR}.

%

In this paper we consider a broader class of evolutionary models, the
\textit{(discrete-time)} Markov models on phylogenetic trees. Briefly, the parameters of these
models consist of a rooted tree topology, a root distribution, and substitution matrices  $P_e$ on the edges $e$ of the tree whose entries correspond to the
conditional probabilities $P(x|y,e)$ that a nucleotide $y$ at the
parent node of $e$ is substituted by nucleotide $x$ at the child
node.
In particular, there is no instantaneous rate matrix fixed for the
whole tree in these models, so that they account for what is
called \textit{nonhomogeneous} data: different lineages in the tree are
allowed to evolve at different rates. We refer to  \cite{Greuel2003}, \cite{Arbook}, and
\cite[chapter 8]{Semple2003} for a mathematical approach to the
evolutionary models used in this paper.

If a DNA sequence has evolved from another according to a substitution matrix $P_e$, then the number of substitutions per site that have occurred can be approximated by
\begin{equation}\label{brlength}
l(e) =-\frac{1}{4}\log \det(P_e) 
\end{equation}
(see \cite{barryhartigan87}). This is usually known as the \textit{branch length} of edge $e$ measured in the expected number of substitutions per site. In the case of stationary continuous-time models, it coincides with $-\frac{1}{4}\tr(D(\Pi)Q \lambda_e t_e)$ if $P_e=\exp(Q\cdot \lambda_e t_e)$ and $D(\Pi)$ is a diagonal matrix with entries corresponding to the stationary distribution $\Pi$.

Generating DNA sequences evolving under a stationary continuous-time
evolutionary model on an edge $e$ with preassigned
branch length $l$ and given rate matrix $Q$, is not difficult: according to equation
\eqref{brlength} one just needs to take $\lambda_e
t_e=-l/\tr(D(\Pi)Q)$ and follow the usual process to generate a
Poisson distribution according to these parameters. There are
several programs available for generating data under most-used
continuous-time evolutionary models, for example \texttt{seq-gen} \cite{Rambaut1997} and
\texttt{evolver} in \textit{PAML} \cite{Yang1997}.

Here we deal with the problem of generating data evolving under the
more general discrete-time models when the branch lengths of the tree are fixed. From what we
have seen above, this problem is equivalent to generate
substitution matrices $P_e$ (belonging to the evolutionary model)
with given determinant. As the substitution matrices are
stochastic matrices, this is not an easy task. We solve this
problem for the so-called equivariant models $\JC$, $\Ka$, $\Kb$ and $\SSM$ (\cite{Draisma},\cite{CFS3}), and for the general Markov model $\GMM$ (\cite{barryhartigan87}, \cite{Steel94}, \cite{Allman2003}). Models $\JC$, $\Ka$, $\Kb$ correspond to the discrete-time version of the corresponding continuous-time models, and $\SSM$ contains $\texttt{HKY}$ as a submodel. Our results for the first four models (Propositions \ref{PropJC}, \ref{PropK80}, \ref{PropK81}, and \ref{prop_SSM}) are
actually bidirectional: we provide algorithms for generating \textit{any}
strictly stochastic matrix $M$ with determinant equal to a given number $K\in (0,1),$ when $M$ is either a $\JC$, $\Ka$, $\Kb$ or $\SSM$ matrix. For the most general model \texttt{GMM} we provide a way of generating strictly stochastic matrices with determinant equal to $K$, but we are not able to claim whether we produce all of them. We observe that we are able to produce matrices that are not a exponential of a real rate matrix (cf. Remark \ref{Rmk_exp}).

The algorithms proposed in this paper have been implemented in C++ in order to generate multiple sequence alignments of DNA data evolving on any phylogenetic tree. This work will be presented in a forthcoming paper. Note that in \cite{Jermiin2003} the authors introduce an  algorithm to generate data on quartet trees under nonhomogeneous continuous-time models.

\section{Preliminaries}

\begin{defi}
A $4\times 4$ matrix $A$ with real entries and row sums equal to
1, $$A=\left( \begin{array}{ccccc}
        a_{1,1}    & a_{1,2} & a_{1,3}    &a_{1,4} \\
        a_{2,1}    & a_{2,2} & a_{2,3}    &a_{2,4} \\
            a_{3,1}    & a_{3,2} & a_{3,3}    &a_{3,4} \\
             a_{4,1}    & a_{4,2} & a_{4,3}    &a_{4,4} \end{array} \right) \quad \left(\sum_ja_{i,j}=1\right),$$
is called a $\GMM$ matrix. The $\GMM$ matrix above is called a
$\SSM$ matrix if $a_{3,1}=a_{2,4},$ $a_{3,2}=a_{2,3},$
$a_{3,3}=a_{2,2},$ $a_{3,4}=a_{2,1},$ $a_{4,1}=a_{1,4},$
$a_{4,2}=a_{1,3},$ $a_{4,3}=a_{1,2},$ $a_{4,4}=a_{1,1}.$ If
moreover $a_{1,1}=a_{2,2},$ $a_{1,2}=a_{2,1},$ $a_{1,3}=a_{2,4}$
and $a_{1,4}=a_{2,3},$ then $A$ is called a $\Ka$ matrix. If a
$\Ka$ matrix satisfies $a_{1,2}=a_{1,4},$ then it is called a
$\Kb$ matrix and it is called a $\JC$ matrix if also
$a_{1,2}=a_{1,3}.$

In other words, a $\SSM$ matrix is a matrix of type
$$\left( \begin{array}{ccccc}
        a    &b & c    &d \\
            e   & f  &g   & h \\
            h  & g  &f   & e \\
             d &   c  &b  &  a   \end{array} \right) \quad \textrm{with
             }\begin{array}{l}
             a+b+c+d=1\\
             e+f+g+h=1
             \end{array};$$
a $\Ka$ matrix is a matrix of type
$$\left( \begin{array}{ccccc}
        a    &b & c    &d \\
            b   & a  &d   & c \\
             c  &  d  &a   & b \\
             d &   c  &b  &  a   \end{array} \right) \quad \textrm{with }a+b+c+d=1;$$
a $\Kb$ matrix is a matrix of type
$$\left( \begin{array}{ccccc}
        a    &b & c    &b \\
            b   & a  &b   & c \\
             c  &  b  &a   & b \\
             b &   c  &b  &  a   \end{array} \right)  \quad \textrm{with }a+2b+c=1;$$
and a $\JC$ matrix is a matrix of type
$$\left( \begin{array}{ccccc}
        a    &b & b    &b \\
            b   & a  &b   & b \\
             b  &  b  &a   & b \\
             b &   b  &b  &  a   \end{array} \right) \quad \textrm{with } a+3b=1.$$
\end{defi}

The names of the matrices above come from well known evolutionary
models: in the stochastic case, $\GMM$ is a transition matrix for
the general Markov model (\cite{barryhartigan87}, \cite{Steel94},
\cite{Allman2003}), $\SSM$ for the strand symmetric model
introduced in \cite{CS}, $\Ka$ for the discrete-time version of
Kimura three-parameters model \cite{Kimura1981}, $\Kb$ for the
discrete-time version of Kimura two-parameters model
\cite{Kimura1980}, and $\JC$ for the discrete-time version of
Jukes-Cantor model \cite{JC69}.

\begin{defi}
A square matrix $A$ is called a  \textit{stochastic matrix} if it
has row sums equal to 1 and nonnegative real entries. It is called
\textit{strictly stochastic} if moreover all its entries are
strictly positive.
\end{defi}

We recall that the determinant of any stochastic matrix has
absolute value  less than or equal to 1 (this is a consequence of Perron-Frobenius theorem). In this paper we address
the problem of providing stochastic matrices of the above shapes
with given determinant $K\in (0,1).$

Before ending the preliminaries section we want to point out in
the lemma below that $\JC$, $\Kb$ and $\Ka$ matrices are
diagonalizable.

\begin{lema}\label{lem_diag} Let $A=\left( \begin{array}{ccccc}
        a    &b & c    &d \\
            b   & a  &d   & c \\
             c  &  d  &a   & b \\
             d &   c  &b  &  a   \end{array} \right)$ be a $\Ka$ matrix ($a+b+c+d=1$) and consider the matrix $$S=\left(
\begin{array}{cccc}
1&1&1&1\\
1&-1&-1&1\\
1 &- 1 & 1& -1\\
1&1 &-1&-1
\end{array} \right).$$  Then $S^{-1}=\frac{1}{4}S$ and
$S^{-1}AS$ is a diagonal matrix with diagonal entries
$\{1,a-b-c+d, a-b+c-d,a+b-c-d\}$ (in this order).
\end{lema}

\begin{rk}
\rm The change of variables considered in the Proposition above corresponds to the discrete Fourier transform in the setting of \cite{Sturmfels2005}.
\end{rk}

\section{Generating $\JC$ matrices with given determinant}

\begin{prop}\label{PropJC}
Let $K \in (0,1)$ and let
$$A=\left( \begin{array}{ccccc}
        a    &b & b    &b \\
            b   & a  &b   & b \\
             b  &  b  &a   & b \\
             b &   b  &b  &  a   \end{array} \right),  \quad a+3b=1,$$
be a $\JC$ matrix. Then $A$ is a strictly stochastic
matrix with determinant equal to $K$ if and only if $a=\frac{1}{4}(1+3K^{1/3})$, $b=\frac{1-a}{3}.$
\end{prop}
\begin{proof} Using Lemma \ref{lem_diag} we have $\det A=(\frac{4a-1}{3})^3.$ Therefore,  $A$ has determinant equal to $K$ if and only if $a=\frac{1}{4}(1+3K^{1/3})$. Moreover, as $K\in (0,1)$, we obtain $1>a>0$ (and so $0<b=\frac{1-a}{3}<1$), and we are done.
\end{proof}
Therefore we have:
\begin{alg}\label{alg_JC}
\rm(Generation of $\JC$ matrices with given determinant.)

\noindent\textit{Input:} $K$ in $(0,1).$

\noindent\textit{Output:} A strictly stochastic $\JC$ matrix $A$
with determinant $K.$
\begin{itemize}
\item[\texttt{Step 1:}] Set $a=\frac{1}{4}(1+3K^{1/3})$, $b=\frac{1-a}{3}.$
\item[\texttt{Final:}] Return
$$A=\left( \begin{array}{ccccc}
        a    &b & b    &b \\
            b   & a  &b   & b \\
             b  &  b  &a   & b \\
             b &   b  &b  &  a   \end{array} \right).$$
\end{itemize}
\end{alg}
\section{Generating $\Kb$ matrices with given determinant}

\begin{rk}\label{uniqueroot}
\rm As a technical step previous to the generation of $\Kb$
matrices with given determinant, we consider the polynomial
$$p_K(x)=-2x^3+x^2+K, \quad K \in (0,1),$$
and we observe that it has exactly one real root $s$ which lies in
$(\sqrt{K},1).$ Indeed, the coefficients of $p_K(x)$ have one
variation in sign and those of $p_K(-x)$ have no variation in
sign. Therefore, applying Descartes' rule we obtain that $p_K(x)$
has exactly one positive root $s$ and no negative roots. Moreover,
as $K$ is a constant in $(0,1)$, we have that
$p_K(\sqrt{K})=2K(1-\sqrt{K})$ is positive and $p_K(1)=K-1$ is
negative, implying that $s$ lies in $(\sqrt{K},1).$

Using the formula for the roots of a cubic polynomial we obtain
$$s=\frac{1}{6}+\frac{1}{6}\sqrt[3]{1+54K+6\sqrt{3K+81K^2}}+\frac{1}{6}\sqrt[3]{1+54K-6\sqrt{3K+81K^2}}.$$

As a byproduct, the polynomial $p_K(-x)$ has exactly one real
root which coincides with $-s.$
\end{rk}

\begin{prop}\label{PropK80}
Let $K \in (0,1)$ and let $s$ be the unique real root of
$p_K(x)=-2x^3+x^2+K$ (see Remark \ref{uniqueroot}). Let
$$A=\left( \begin{array}{ccccc}
        a    &b & c    &b \\
            b   & a  &b   & c \\
             c  &  b  &a   & b \\
             b &   c  &b  &  a   \end{array} \right),$$
be a $\Kb$ matrix ($a+2b+c=1$), and consider the change of variables
$\al=1-2(b+c),$ $\beta=1-4b.$ Then $A$ is a strictly stochastic
matrix with determinant equal to $K$ if and only if $\sqrt{K} <
|\al| < s$ and $\beta = K/\al^2.$
\end{prop}

\begin{proof}
First we note that the inverse change of variables is
$b=\frac{1-\beta}{4},$ $c=\frac{1+\beta-2\al}{4}.$ Moreover, $\al=a-c$ and $\beta=a-2b+c$ are the diagonal entries in $S^{-1}AS$ (different than 1) in Lemma \ref{lem_diag} and therefore $\det(A)=\al^2 \beta.$

$\Rightarrow$) Assume that $A$ is strictly stochastic with
determinant $K.$ Then $b$ is strictly positive, so that $\beta
<1.$ As $K=\det(A)=\al^2 \beta$ and $\beta <1$, we obtain $|\al|
>\sqrt{K}$. In particular, $\al \neq 0$ and we can
write $\beta =K/\al^2.$

Using the inverse change of variables above and $\beta =K/\al^2$
we have
$$a>0 \Leftrightarrow 2b+c<1 \Leftrightarrow \frac{3-K/\al^2 -2\al}{4}<1\Leftrightarrow p_K(-\al)>0.$$
As noted in Remark \ref{uniqueroot}, $p_K(-x)$ has exactly one
negative root which equals $-s $ and lies in $(-1,-\sqrt{K}).$ As
$p_K(-x)$ has positive leading term, $p_K(-\al)>0$ only holds if
$\al
>-s.$

Similarly, $c$ is strictly positive if and only if $p_K(\al)>0.$
Following an analogous argument, we obtain that $p_K(\al)>0$ if and
only if $\al <s$. Putting all together we obtain
$\sqrt{K}<|\al|<s,$ as desired.

$\Leftarrow$) Assume that $\sqrt{K} < |\al| < s$ and $\beta =
K/\al^2.$ In particular, we have $< \beta < \frac{K}{K}=1$ and we obtaing that
$b=\frac{1-\beta}{4}$ is strictly positive.

Now, as in the proof of $\Rightarrow$) we have that $c>0$ if and
only if $p_K(\al)>0.$ And also as above, this happens if and only
if $\al <s.$ As we assumed $|\al|<s$, we obtain $c>0.$

Lastly, $a>0$ if and only of $p_K(-\al)>0$, and this holds if and
only if $\al >-s$ (see proof of $\Rightarrow$). As we assumed
$|\al|<s$, we get that $A$ is a strictly stochastic matrix.

Moreover, $\det(A)=\al^2\beta=K$ as wanted.\end{proof}

Using the previous result, we provide the following algorithm for
generating strictly stochastic $\Kb$ matrices with given
determinant $K$. It is worth pointing out that with this algorithm
we are generating \textit{all} $\Kb$ strictly stochastic matrices
with determinant $K$.

\begin{alg}\label{alg_Kb}
\rm(Generation of $\Kb$ matrices with given determinant.)

\noindent\textit{Input:} $K$ in $(0,1).$

\noindent\textit{Output:} A strictly stochastic $\Kb$ matrix $A$
with determinant $K.$
\begin{itemize}
\item[\texttt{Step 1:}] Compute the unique real root $s$ of
$p_K(x)$ using Remark \ref{uniqueroot}.
\item[\texttt{Step 2:}] Choose $\al$ randomly  such that $\sqrt{K}<|\al|<s.$
\item[\texttt{Step 3:}] Let $\beta:=K/\al^2$, $b:=\frac{1-\beta}{4},$
$c:=\frac{1+\beta-2\al}{4},$ and $a:=1-2b-c.$
\item[\texttt{Final:}] Return $$A:=\left( \begin{array}{ccccc}
        a    &b & c    &b \\
            b   & a  &b   & c \\
             c  &  b  &a   & b \\
             b &   c  &b  &  a   \end{array} \right).$$
\end{itemize}
\end{alg}

\section{Generating $\Ka$ matrices with given determinant}
Previously to dealing with the case of $\Ka$ matrices, for each real number $K$ in $(0,1)$, we let $s$ be the unique positive root of the polynomial
$$q_K(z):=z(z+1)^2-4K.$$
Indeed, according to Descartes' rules of signs, this polynomial
has at most one positive root. Moreover, as $q_K(K)<0$ and
$q_K(1)>0$, there is exactly one positive root $s$ and it lies in
$(K,1).$ Using the formula for the roots of a cubic polynomial we
obtain
\begin{equation}\label{sol_cubic}
s=-\frac{2}{3}-\frac{1}{3}\sqrt[3]{-1-54K+6\sqrt{3K+81K^2}}-\frac{1}{3}\sqrt[3]{-1-54K-6\sqrt{3K+81K^2}}.
\end{equation}

\begin{prop}\label{PropK81}
Let $K \in (0,1)$ and let $s$ be the unique real root of
$q_K(z):=z(z+1)^2-4K.$ Let
$$A=\left( \begin{array}{ccccc}
        a    &b & c    &d \\
            b   & a  & d   & c \\
             c  & d  &a   & b \\
            d &   c  &b  &  a   \end{array} \right),$$
be a $\Ka$ matrix ($a+2b+c=1$), and consider the change of variables
$\alpha=1-2(b+c),$ $\beta=1-2(b+d),$ $\gamma=1-2(c+d).$ Then $A$ is a strictly stochastic
matrix with determinant equal to $K$ if and only if $|\alpha|\in (s,1)$, $|\beta| \in \left(I_{|\al|},J_{|\al|}\right)$ where
$$I_{|\al|}=\max\left\{\frac{-1+{|\al|}+\sqrt{(1-{|\al|})^2+\frac{4K}{{|\al|}}}}{2},\frac{1+{|\al|}-\sqrt{(1+{|\al|})^2-\frac{4K}{{|\al|}}}}{2}\right\},$$
$$J_{|\al|}=\min\left\{\frac{1+{|\al|}+\sqrt{(1+{|\al|})^2-\frac{4K}{{|\al|}}}}{2},\frac{1-{|\al|}+\sqrt{(1-{|\al|})^2+\frac{4K}{|\al|}}}{2}\right\},$$
and $\gamma = \frac{K}{\al \beta}.$
\end{prop}

\begin{rk}
\rm
As the change of variables above is symmetric in $b,c,d$, the roles of these three variables can be exchanged in the previous Proposition.
\end{rk}

Before proving this Proposition we need the following technical lemma.

\begin{lema}\label{Lem_K81}
Let $K$ be a real number in $(0,1)$, let $s$ be the unique
positive solution to $z(z+1)^2-4K=0$,  and consider the function
$$f(x,y)=1-x-y+\frac{K}{xy}$$
 defined over $\RR^2 \smallsetminus \{0\}.$%
Given $y>0$, we consider the set
$$\Omega_y=\left\{x\in \RR \, \left| \, x>0, f(x,y)>0, f(x,-y)>0, f(-x,y)>0, f(-x,-y)>0\right. \right\}.$$
Then $\Omega_y$ is not empty if and only if $y>s.$ Moreover, if $x \in \Omega_y$ and $y<1$, then $x$ belongs to $\left(I_y,J_y\right)$ where
$$I_y=\max\left\{\frac{-1+y+\sqrt{(1-y)^2+\frac{4K}{y}}}{2},\frac{1+y-\sqrt{(1-y)^2-\frac{4K}{y}}}{2}\right\} \quad \textrm{and}$$ $$J_y=\min\left\{\frac{1+y+\sqrt{(1-y)^2-\frac{4K}{y}}}{2},\frac{1-y+\sqrt{(1-y)^2+\frac{4K}{y}}}{2}\right\}.$$
\end{lema}

\begin{proof}
We fix $y>0$, and we view $f$ and $g$ as functions on $x$. For
$x>0$ we can multiply $f$, $g$ by $x$ and define quadratic
functions $\tilde{f}_y(x):=-x^2+(1-y)x+K/y $ and
$\tilde{g}_y(x):=x^2+(1+y)x+K/y$ so that $x$ belongs to
$\Omega_y$ if and only if $x>0$, $\tilde{f}_y(x)>0$,
$\tilde{f}_{-y}(x)>0$, $\tilde{g}_y(x)>0$ and
$\tilde{g}_{-y}(x)>0.$

Note that $\tilde{f}_y$ has discriminant
$\Delta_1(y)=(1-y)^2+\frac{4K}{y}$ and $\tilde{g}_y$ has
discriminant $\Delta_2(y)=(1+y)^2-\frac{4K}{y}.$

We observe that $\Delta_1(y)>0$ for $y>0$. Therefore
$\tilde{f}_y(x)=0$ has two real solutions
$x_{1,L}(y)=\frac{1-y-\sqrt{\Delta_1(y)}}{2}$,
$x_{1,R}(y)=\frac{1-y+\sqrt{\Delta_1(y)}}{2}$, and
$\tilde{f}_y(x)$ is positive for $x$ in $(x_{1,L},x_{1,R}).$ Note
that $\sqrt{\Delta_1(y)}>|1-y|$ for $y>0$, so $x_{1,L}(y)$ is
negative and $x_{1,R}(y)$ is positive. Therefore, for $x>0$ and
$y>0$, $\tilde{f}_y(x)$ is positive if and only if $x \in
(0,x_{1,R}(y)).$

On the other hand, as $\tilde{f}_{-y}$ has negative leading
coefficient, there exists $x$ with $\tilde{f}_{-y}(x)>0$ if and
only if $\Delta_1(-y)>0$. Note that $\Delta_1(-y)$ is positive for
$y>0$ if and only if $y>s$ (indeed, $\Delta_1(-y)$ coincides with $q_K(y)/y$).

Thus $\tilde{f}_{-y}(x)>0$ has a solution for $x>0$,  if and only
if $y>s.$ Now for $x>0,y>s$, the roots of $\tilde{f}_{-y}(x)=0$
are $x_{1,L}(-y)$ and $x_{1,R}(-y).$ Clearly $x_{1,R}(-y)$ and
$x_{1,L}(-y)$ are both positive for $y>s.$ Therefore, for $x>0$ and
$y>0$, we have $ \tilde{f}_{-y}(x)>0$ if and only if $y>s$ and $x\in
(x_{1,L}(-y),x_{1,R}(-y))$.

Now we study the positivity of $\tilde{g}_{y}(x)$ for $x>0$. Note
that $\tilde{g}_{y}$ has discriminant $\Delta_1(-y)$. As the
leading coefficient of $\tilde{g}_{y}$ is positive, we have that
$\tilde{g}_{y}(x)>0$ for all $y<s$ and $x\in \RR$ (because in this case the discriminant is negative). Moreover, if
$y>s$, the real roots of $\tilde{g}_{y}(x)=0$ are
$x_{2,L}(y)=\frac{-(1+y)-\sqrt{\Delta_1(-y)}}{2}$ and
$x_{2,R}(y)=\frac{-(1+y)+\sqrt{\Delta_1(-y)}}{2}.$ They are both
negative so that $\tilde{g}_{y}(-x)$ is positive for all $y>s$
and $x>0.$

We study the positivity of $\tilde{g}_{-y}(x)$ for $x>0$ and
$y>0$. The discriminant of $\tilde{g}_{-y}$ is $\Delta_1(y)$, and
it is positive for $y>0.$ Then the roots of $\tilde{g}_{-y}$ are
$x_{2,L}(-y)$ and $x_{2,R}(-y)$. For $y>0$ we have $x_{2,L}(-y)<0$
and $x_{2,R}(-y)>0$, and therefore $\tilde{g}_{-y}(x)>0$ if and
only if $x$ belongs to $(x_{2,R}(-y),+\infty).$

Summing up, we have proven that the set $\Omega_y$ is non-empty if
and only if $y>s.$ Moreover, in that case, if $x$ belongs to
$\Omega_y$, then $x$ lies in
\begin{equation*}
\left(0,x_{1,R}(y)\right)\cap \left(x_{1,L}(-y),x_{1,R}(-y)\right) \cap \left(0,+\infty\right) \cap \left(x_{2,R}(-y),+\infty\right).
\end{equation*}

It is easy to see that $x_{1,R}(y)$ is bigger than $x_{2,R}(-y)$ for $y>0$. Therefore the intersection of intervals above is equal to
$$\left(x_{1,L}(-y),x_{1,R}(-y)\right) \cap \left(x_{2,R}(-y),x_{1,R}(y)\right).$$

The statement of the lemma follows from the following claim.

\textit{Claim:} If $y<1$, then $x_{2,R}(-y)<x_{1,R}(-y).$

\textit{Proof of Claim:} This is equivalent to proving
\begin{equation}\label{eq_ineq}
\sqrt{\Delta_1(y)}-\sqrt{\Delta_1(-y)}<2.
\end{equation}

First of all we note that $\Delta_1(y) \leq \Delta_1(-y)$ if and only if $y\geq \frac{2K}{y}$. As $y>0$, this holds if and only if $y\geq \sqrt{2K}$. Therefore, for $y\geq \sqrt{2K}$, $\sqrt{\Delta_1(y)}-\sqrt{\Delta_1(-y)}$ is negative (and hence $<2.$)

If $y<\sqrt{2K}$, we have just seen that $\sqrt{\Delta_1(y)}>\sqrt{\Delta_1(-y)}.$ In this case, both sides in \eqref{eq_ineq} are positive and hence it is equivalent when raising it to the second power:
$$\Delta_1(y)+\Delta_1(-y)-2\sqrt{\Delta_1(y)\Delta_1(-y)}<4.$$

As we are assuming $y<1$, we have $\Delta_1(y)+\Delta_1(-y)-4=2y^2-2<0<2\sqrt{\Delta_1(y)\Delta_1(-y)},$ as we wanted to prove.
\end{proof}

\begin{proofPropK81}
Taking into account that $a=1-b-c-d$, we note that inverse change of variables is
$a=\frac{1}{4}(1+\al +\beta+ \gamma)$, $b=\frac{1}{4}(1-\al -\beta+ \gamma)$, $c=\frac{1}{4}(1-\al +\beta- \gamma)$, $d=\frac{1}{4}(1+\al -\beta- \gamma).$ Observing that $\al,\beta,\gamma$ are the diagonal entries in $S^{-1}AS$ in Lemma \ref{lem_diag}, we see that $\det A=\al \beta\gamma.$

$\Rightarrow$) Assume that $A$ is stochastic with determinant $K \in \left(0, 1\right)$. Then $\al$, $\beta,$ and $\gamma$ are non-zero, and $\gamma=\frac{K}{\al \beta}.$ From the positivity of $a,b,c,d$ we get that
$1+\al +\beta+ \frac{K}{\al \beta}>0$, $1-\al -\beta+\frac{K}{\al \beta}>0$, $1-\al +\beta-\frac{K}{\al \beta}>0,$ and $1+\al -\beta-\frac{K}{\al \beta}>0.$
In terms of Lemma \ref{Lem_K81}, these inequalities can be rewritten as
$$f(-\beta,-\al)>0, f(\beta,\al)>0, f(\beta,-\al)>0, f(-\beta,\al)>0. $$
Therefore $|\beta|$ is an element of $\Omega_{|\al|}$, which implies that $|\al|>s$ (see Lemma \ref{Lem_K81}).
Moreover, as $\al=1-2(b+d)$, and $b,d>0$, we see that $|\al|<1.$ The result then follows from Lemma \ref{Lem_K81}.

$\Leftarrow$) Using Lemma \ref{Lem_K81} we see that under these assumptions, $\Omega_{|\al|}\neq \emptyset$ and $|\beta|$ belongs to $\Omega_{|\al|}.$ Therefore
$f(-\beta,-\al)>0, f(\beta,\al)>0, f(\beta,-\al)>0, f(-\beta,\al)>0.$ As $\gamma=\frac{K}{\al\beta}$, these inequalities coincide with $a>0$, $b>0$, $c>0$ and $d>0$, and we are done.
\end{proofPropK81}

The previous results give us a way of generating \textit{any} $\Ka$ matrix.
\begin{alg}\rm(Generation of $\Ka$ matrices with given determinant.)

\noindent\textit{Input:} $K$ in $(0,1).$

\noindent\textit{Output:} A strictly stochastic $\Ka$ matrix $A$
with determinant $K.$
\begin{itemize}
\item[\texttt{Step 1:}] Compute the unique real root $s$ of
$z(z+1)^2-4K$using \eqref{sol_cubic}.
\item[\texttt{Step 2:}] Choose $\al$ randomly  such that $1>|\al|>s.$
\item[\texttt{Step 3:}] Take $\beta$ randomly such that $|\beta|$ belongs to $(I_{|\al|},J_{|\al|})$.
\item[\texttt{Step 4:}] Set $\gamma=\frac{K}{\al\beta}.$
\item[\texttt{Step 5:}] Set $a=\frac{1}{4}(1+\al +\beta+ \gamma)$, $b=\frac{1}{4}(1-\al -\beta+ \gamma)$, $c=\frac{1}{4}(1-\al +\beta- \gamma)$, $d=\frac{1}{4}(1+\al -\beta- \gamma).$
\item[\texttt{Final:}] Return $$A=\left( \begin{array}{ccccc}
        a    &b & c    &d \\
            b   & a  & d   & c \\
             c  & d  &a   & b \\
            d &   c  &b  &  a   \end{array} \right).$$
\end{itemize}
\end{alg}

\begin{rk}
\label{Rmk_exp}\rm
The change of variables in Proposition \ref{PropK81} diagonalizes the matrix to Diag$(1,\al,\beta,\gamma)$ (see Lemma \ref{lem_diag}). As we have seen in that proposition, $\al$ and $\beta$ can be both negative. Therefore, using \cite{Culver}, we observe that the matrices produced by the algorithm above are not all of them of type $\exp(Q)$ for a real matrix $Q.$
\end{rk}

\section{Generating $\SSM$ matrices with given determinant}

\begin{defi}
Let $A$ be a $4\times4$ real matrix. We call $F(A)$ the matrix
obtained from $A$ after performing the basis change
$F(A)=S^{-1}AS$ where
$$S=\left(
\begin{array}{cccc}
1&0&0&-1\\
0&1&1&0\\
0 & 1 & -1& 0\\
1&0 &0&1
\end{array} \right).
$$
\end{defi}

When $A$ is a $\SSM$ matrix, $A$ can be viewed as an element in
$Hom_G(\CC^4,\CC^4)$ where $G=<(\a\t)(\c\g)>$ (see \cite{CFS}).
The change of basis above decomposes $\CC^4$ into its isotypic
components via the natural linear representation $G\lra
GL(\CC^4).$ This change of basis is also known as the generalized
Fourier transform (see \cite{CS}). We have the following fact:

\begin{lema}\label{rem_transform}
A $4\times 4$ matrix $A=(a_{i,j})$ is a $\SSM$ matrix if and only
if $F(A)$ has the following shape:
$$
F(A)=\left( \begin{array}{ccccc}
        \lambda   & 1-\lambda&0    &0 \\
            1-\mu & \mu  & 0   & 0 \\
             0  &  0  &\al  & \al' \\
             0 &   0  &\beta'  &  \beta   \end{array} \right).
$$
In this case, $\lambda,\mu$, $\al,\al'$, $\beta,\beta'$ can be written in terms
of the entries of $A$ as $\lambda = a_{1,1}+a_{1,4}$,
$\mu=a_{2,2}+a_{2,3}$, $\al=a_{2,2}-a_{2,3},$ $\al'=a_{2,4}-a_{2,1}$, $\beta=a_{1,1}-a_{1,4}$, and
$\beta'=a_{1,3}-a_{1,2}$.
 The inverse change of variables is
$a_{1,1}=(\lambda+\beta)/2$, $a_{1,2}=(1-\lambda-\beta')/2$,
$a_{1,3}=(1-\lambda+\beta')/2$ $a_{1,4}=(\lambda-\beta)/2$,
$a_{2,1}=(1-\mu-\al')/2$, $a_{2,2}=(\mu+\al)/2$,
$a_{2,3}=(\mu-\al)/2,$ $a_{2,4}=(1-\mu+\al')/2.$
\end{lema}
\begin{proof}
The matrix $F(A)$ for a generic matrix $A=(a_{i,j})$ is
$$\frac{1}{2}{\scriptsize\left(\begin{array}{cc|cc}
a_{1,1}+a_{1,4}+a_{4,1}+a_{4,4} & a_{1,2}+a_{1,3}+a_{4,2}+a_{4,3} &
a_{1,2}-a_{1,3}+a_{4,2}-a_{4,3} & a_{1,4}-a_{1,1}-a_{4,1}+a_{4,4}\\
a_{2,1}+a_{2,4}+a_{3,1}+a_{3,4} & a_{2,2}+a_{2,3}+a_{3,2}+a_{3,3} &
a_{2,2}-a_{2,3}+a_{3,2}-a_{3,3} & a_{2,4}-a_{2,1}-a_{3,1}+a_{3,4}\\
\hline
a_{2,1}+a_{2,4}-a_{3,1}-a_{3,4} & a_{2,2}+a_{2,3}-a_{3,2}-a_{3,3} &
a_{2,2}-a_{2,3}-a_{3,2}+a_{3,3} & a_{2,4}-a_{2,1}+a_{3,1}-a_{3,4}\\
a_{4,1}+a_{4,4}-a_{1,1}-a_{1,4} & a_{4,2}+a_{4,3}-a_{1,2}-a_{1,3} &
a_{1,3}-a_{1,2}+a_{4,2}-a_{4,3} & a_{1,1}-a_{1,4}-a_{4,1}+a_{4,4}
\end{array}\right).}
$$

If $A$ is a $\SSM$ matrix, then $a_{3,1}=a_{2,4},$
$a_{3,2}=a_{2,3},$ $a_{3,3}=a_{2,2},$ $a_{3,4}=a_{2,1},$
$a_{4,1}=a_{1,4},$ $a_{4,2}=a_{1,3},$ $a_{4,3}=a_{1,2},$ and
$a_{4,4}=a_{1,1}.$ Therefore the non-diagonal blocks are 0. Moreover,
as sums of rows are equal to 1, we have that the entries of each
row in the upper left block sum to 1:
$$\frac{1}{2}(a_{1,1}+a_{1,4}+a_{4,1}+a_{4,4} + a_{1,2}+a_{1,3}+a_{4,2}+a_{4,3})=1,$$
$$\frac{1}{2}(a_{2,1}+a_{2,4}+a_{3,1}+a_{3,4} + a_{2,2}+a_{2,3}+a_{3,2}+a_{3,3})=1.$$

Conversely, imposing that the entries of non-diagonal blocks in
$F(A)$ are equal to 0 is equivalent to imposing $a_{3,1}=a_{2,4},$
$a_{3,2}=a_{2,3},$ $a_{3,3}=a_{2,2},$ $a_{3,4}=a_{2,1},$
$a_{4,1}=a_{1,4},$ $a_{4,2}=a_{1,3},$ $a_{4,3}=a_{1,2},$ and
$a_{4,4}=a_{1,1}$ (adding and
subtracting certain pairs of equations).  Moreover, $F(A)_{1,1}+F(A)_{1,2}=1$ implies
that sum of rows 1 and 4 is equal to 2 (and similar for rows 2 and
3). But we have just seen that the set of entries in the first
(resp. second) row is equal to the set of entries in the forth
(resp. third) row, thus the sum of entries in each row is equal to
1.
\end{proof}

In the following lemma we characterize the stochasticity of $A$
via $F(A)$.
\begin{lema}\label{lem_SSM}
$A$ is a strictly stochastic $\SSM$ matrix if and only if
$$
F(A)=\left( \begin{array}{ccccc}
        \lambda   & 1-\lambda&0    &0 \\
            1-\mu & \mu  & 0   & 0 \\
             0  &  0  &\al  & \al' \\
             0 &   0  &\beta'  &  \beta   \end{array} \right)
$$
with $\lambda,\mu \in (0,1)$, $|\beta|<\lambda$, $|\beta'|<1-\lambda$,
$|\al|<\mu$, and $|\al'|<1-\mu$.
\end{lema}

\begin{proof}
If $A$ is a $\SSM$ matrix, then
$$A=\left(
\begin{array}{ccccc}
        a    &b & c    &d \\
            e   & f  &g   & h \\
            h  & g  &f   & e \\
             d &   c  &b  &  a   \end{array} \right)$$ with
$a+b+c+d=1$, $e+f+g+h=1$,
              and by Lemma \ref{rem_transform}, $F(A)$ has the shape above
with $\lambda = a+d$, $\mu=g+f$, $\beta=a-d$, $\beta'=c-b$, $\al=f-g$, and $\al'=h-e$.

If $a,b,\dots,h$ are strictly positive, then we clearly have
$\lambda,\mu \in (0,1)$, $|\al|<\mu$,$|\al'|<1-\mu$, $|\beta|<\lambda$, and $|\beta'|<1-\lambda$. 

Conversely, if $F(A)$ is block-diagonal as in the statement of the
lemma, we know by Lemma \ref{rem_transform} that $A$ is a $\SSM$
matrix with entries as above. As the inverse change of variables
is $a=(\lambda+\beta)/2$, $b=(1-\lambda-\beta')/2$, $c=(1-\lambda+\beta')/2$
$d=(\lambda-\beta)/2$, $e=(1-\mu-\al')/2$, $f=(\mu+\al)/2$,
$g=(\mu-\al)/2,$ $h=(1-\mu+\al')/2$, then if $\lambda,\mu$ lie $(0,1)$, $|\al|<\mu$, $|\al'|<1-\mu$,
$|\beta|<\lambda$, and $|\beta'|<1-\lambda$,  we
obtain that $a,b, \dots,h$ are strictly positive.
\end{proof}

Before stating the main result of this section we introduce some notation and we prove a technical result.

\begin{rk}\label{rmk_SSM}
\rm
Given $K \in (0,1)$, we consider the polynomial $r_K(z)=z^3+z-2K.$ It has a unique positive real root.  Indeed, by
Descartes' rule of signs we see that $r_K$ has at most one
positive real root. Moreover, as $r_K(K)$  is strictly negative
and $r_K(1)$ is strictly positive, there exists exactly one
positive root $\nu_0$ of $r_K(z)$ and it lies in $(K,1)$. Using
the formula for the roots of a cubic polynomial we actually get
$$\nu_0=-\frac{1}{3}\sqrt[3]{-27K+3\sqrt{81K^2+3}}-\frac{1}{3}\sqrt[3]{-27K-3\sqrt{81K^2+3}}.$$
\end{rk}

\begin{defi}
Given $K \in (0,1)$, we consider the polynomial $r_K(z)=z^3+z-2K$
and we call $\nu_0$ its unique positive root (Remark \ref{rmk_SSM}). We define $\Theta$ as the set of points $(\lambda,\mu)\in (0,1)^2$
satisfying
$$\nu_0+1\leq \lambda+\mu<2, \,\textrm{ and } |\lambda-\mu|<\min\left\{2-\lambda-\mu,\sqrt{\frac{r_K(\lambda+\mu-1)}{\lambda+\mu-1}}\right\}.$$
\end{defi}

\begin{lema}\label{lem_ineq}
Let $\lambda,\mu$ be real numbers in $(0,1)$ with $\lambda+\mu>1.$ Then
$(\lambda,\mu)$ belongs to $\Theta$ if and only if
\begin{equation}\label{eq_lemaSSM}
\frac{K}{\lambda+\mu-1}-(1-\lambda)(1-\mu)<\lambda\mu.
\end{equation}
\end{lema}
\begin{proof}
As $\lambda+\mu>1$, we exchange the inequality \eqref{eq_lemaSSM} by
the following equivalent inequality:
\begin{equation}\label{ineq_SSM}
(\lambda+\mu-1)(2\lambda\mu+1-\lambda-\mu)-K>0.
\end{equation}

We consider the change of variables $s:=\lambda+\mu$, $t:=\lambda-\mu$ (so
that $\lambda=\frac{s+t}{2}$, $\mu=\frac{s-t}{2}$). We observe that
$\lambda$ and $\mu$ lie in $(0,1)$ if and only if $|t|<s$ and
$|t|<2-s.$ As we are assuming $\lambda+\mu>1$, we have $s>2-s$.
Therefore, $\lambda,\mu$ are real numbers in $(0,1)$ with $\lambda+\mu>1$
if and only if $|t|<2-s.$

In these new variables inequality \eqref{ineq_SSM} reads as
$(s-1)(\frac{s^2-t^2}{2}+1-s)-K>0,$ which is equivalent to
\begin{equation}\label{ineq_st}
t^2<\frac{(s-1)((s-1)^2+1)-2K}{s-1}=\frac{r_K(s-1)}{s-1}.
\end{equation}

$\Leftarrow$) Let $\lambda,\mu$ be real numbers in $(0,1)$ satisfying
$\lambda+\mu>1$ and \eqref{ineq_SSM}. Then $s:=\lambda+\mu$ lies
in $(1,2)$, $|t:=\lambda-\mu|<2-s$, and $s,t$ satisfy
\eqref{ineq_st}. In particular, $\frac{r_K(s-1)}{s-1} \geq 0$. As
we have $s>1$, this inequality is positive if and only if its
numerator is positive, which holds if and only if $s-1\geq \nu_0.$
Therefore $s$ is in $[\nu_0+1,2)$ and
$|t|<\min\left\{2-s,\sqrt{\frac{r_K(s-1)}{s-1}}\right\}$; in other
words, $(\lambda,\mu)$ belongs to $\Theta.$

$\Rightarrow$) Conversely, let $(\lambda,\mu) \in \Theta.$
Then, using the change of variables above, we have that $(s,t)$ satisfies
$|t|<\sqrt{\frac{r_K(s-1)}{s-1}}$. In particular, \eqref{ineq_st} is satisfied and hence \eqref{eq_lemaSSM} is satisfied as well.
\end{proof}
%

\begin{prop}\label{prop_SSM}
Given $K$ a real number in $(0,1)$, we consider the polynomial
$r_K(z)=z^3+z-2K$ and let $\nu_0$ be its positive real root in
$(K,1)$ (see Remark \ref{rmk_SSM}). We fix two real numbers
$\lambda,\mu$ in $(0,1)$ such that $\lambda+\mu>1$. Then the set
$$\Omega_{\lambda,\mu}=\left\{(\al,\beta) \in \RR^2 \left|  0<\al<\mu, |\beta|<\lambda, |\al\beta-\frac{K}{\lambda+\mu-1}|<(1-\lambda)(1-\mu)\right.\right\}$$
is non-empty if and only if $(\lambda,\mu)$ belongs to $\Theta.$
Moreover in this case, $(\al,\beta)$ belongs to $\Omega_{\lambda,\mu}$
if and only if $\al$ belongs to
$\left(\frac{\frac{K}{\lambda+\mu-1}-(1-\lambda)(1-\mu)}{\lambda},\mu\right),$ $\al>0,$
and
$$\max\left\{-\lambda, \frac{\frac{K}{\lambda+\mu-1}-(1-\lambda)(1-\mu)}{\al}\right\} <\beta<
\min\left\{\lambda,
\frac{\frac{K}{\lambda+\mu-1}+(1-\lambda)(1-\mu)}{\al}\right\}.$$
\end{prop}

\begin{proof}

$\Rightarrow$) If $(\al,\beta)$ is a point in $\Omega_{\lambda,\mu}$, then $|\al\beta-\frac{K}{\lambda+\mu-1}|<(1-\lambda)(1-\mu).$  This is equivalent to
\begin{equation}
\label{ineq_y1y4}
\frac{K}{\lambda+\mu-1}-(1-\lambda)(1-\mu)<\al\beta<\frac{K}{\lambda+\mu-1}+(1-\lambda)(1-\mu).
\end{equation}

In particular, as $\al\beta<\lambda\mu$, we have
$$\frac{K}{\lambda+\mu-1}-(1-\lambda)(1-\mu)<\lambda\mu.
$$ Hence, using Lemma \ref{lem_ineq} we obtain $(\lambda,\mu)\in \Theta$.

Moreover, as $|\beta|<\lambda$, inequality
$\frac{K}{\lambda+\mu-1}-(1-\lambda)(1-\mu)<\al\beta$ implies
$\frac{K}{\lambda+\mu-1}-(1-\lambda)(1-\mu)<\lambda\al,$ and therefore $\al$
belongs to the interval
$$\left(\frac{\frac{K}{\lambda+\mu-1}-(1-\lambda)(1-\mu)}{\lambda},\mu\right).$$ The
inequalities on $\beta$ follow directly from \eqref{ineq_y1y4} and from $|\beta|<\lambda$. Conversely, if $\al$
belongs to the above interval, and $\beta$ satisfies
$$\max\left\{-\lambda, \frac{\frac{K}{\lambda+\mu-1}-(1-\lambda)(1-\mu)}{\al}\right\} <\beta<
\min\left\{\lambda,
\frac{\frac{K}{\lambda+\mu-1}+(1-\lambda)(1-\mu)}{\al}\right\},$$ then inequalities \eqref{ineq_y1y4} hold
and hence $(\al,\beta)$ lies in $\Omega_{\lambda,\mu}.$

$\Leftarrow$) Let $(\lambda,\mu)$ be a point in $\Theta.$ In this case
$(\lambda,\mu)$ satisfies \eqref{eq_lemaSSM}, and in particular, the
interval
\begin{equation}\label{int_y4}
\left(\frac{\frac{K}{\lambda+\mu-1}-(1-\lambda)(1-\mu)}{\lambda},\mu\right)
\end{equation}
is non-empty. We choose $\al>0$ in this interval.

Then, the interval $$\left(
\frac{\frac{K}{\lambda+\mu-1}-(1-\lambda)(1-\mu)}{\al},\frac{\frac{K}{\lambda+\mu-1}+(1-\lambda)(1-\mu)}{\al}\right)$$
is non-empty (the left-hand side numerator is smaller than the right-hand side numerator, and the denominator is positive) and its
intersection with $(-\lambda,\lambda)$ is not empty. Indeed, as $\al>0$
and $\al$ belongs to the interval \eqref{int_y4}, we have
$$\frac{\frac{K}{\lambda+\mu-1}-(1-\lambda)(1-\mu)}{\al}<\lambda;$$ moreover $-\lambda$ is less than
$\frac{\frac{K}{\lambda+\mu-1}+(1-\lambda)(1-\mu)}{\al}$
because this expression is positive.

Finally, we choose $\beta$ in this intersection of intervals and we
obtain a point $(\al,\beta)$ in $\Omega_{\lambda,\mu}.$
%
\end{proof}

\begin{thm}\label{thm_SSM}
Let $K$ be a real number in $(0,1)$.
\begin{enumerate}
\item[(a)] Let $(\lambda,\mu)$ be a point in $\Theta$,
let $(\al,\beta)$ be a point in $\Omega_{\lambda,\mu}$, and consider real numbers
$\al'$ and $\beta'$ such that

\begin{enumerate}
\item[(i)] $\frac{|\al\beta-\frac{K}{\lambda+\mu-1}|}{1-\mu}<|\beta'|
<1-\lambda,$ and
\item[(ii)] $\al'=\frac{\al\beta-\frac{K}{\lambda+\mu-1}}{\beta'}.$
\end{enumerate}
Then, if we consider the change of variables
$a=(\lambda+\beta)/2$,$b=(1-\lambda-\beta')/2$,$c=(1-\lambda+\beta')/2$
$d=(\lambda-\beta)/2$, $e=(1-\mu-\al')/2$, $f=(\mu+\al)/2$,
$g=(\mu-\al)/2,$  $h=(1-\mu+\al')/2,$ the matrix $$A=\left(
\begin{array}{ccccc}
        a    &b & c    &d \\
            e   & f  &g   & h \\
            h  & g  &f   & e \\
             d &   c  &b  &  a   \end{array} \right)$$
is  a strictly stochastic $\SSM$ matrix with determinant $K$, $a+d+f+g>1$, $b\neq c$, and $f<g$.
\item[(b)] Conversely, let $$A=\left( \begin{array}{ccccc}
        a    &b & c    &d \\
            e   & f  &g   & h \\
            h  & g  &f   & e \\
             d &   c  &b  &  a   \end{array} \right)$$
be a strictly stochastic $\SSM$ matrix with determinant $K$ and
with $a+d+g+f>1$, $b\neq c$ and $f>g$. Then $F(A)$ is equal to
$$
\left( \begin{array}{ccccc}
        \lambda   & 1-\lambda&0    &0 \\
            1-\mu & \mu  & 0   & 0 \\
             0  &  0  &\al  & \al' \\
             0 &   0  &\beta'  &  \beta   \end{array} \right),
$$
where $(\lambda,\mu) \in \Theta$, $(\al,\beta) \in \Omega_{\lambda,\mu}$, and $\al'$, $\beta'$  satisfy
conditions (i) and (ii) stated in $(a).$
\end{enumerate}
\end{thm}

\begin{rk}\rm
(1) By Proposition \ref{prop_SSM}, if $(\lambda,\mu)$ is a point in $\Theta$, there exists $(\al,\beta) \in \Omega_{\lambda,\mu}.$ This implies that $|\al\beta-\frac{K}{\lambda+\mu-1}|$ is smaller than $(1-\lambda)(1-\mu)$, and thus the interval
$$\left(\frac{|\al\beta-\frac{K}{\lambda+\mu-1}|}{1-\mu},1-\lambda\right)$$
is non-empty. In particular, there exists $\beta'$ in this interval. Therefore conditions (i) and (ii) in Theorem \ref{thm_SSM}(a) are not empty.

(2) Assumptions $a+d+g+f>1$, $f>g$, $b\neq c$ are biologically
meaningful: the elements in the diagonal of an evolutionary Markov
matrix stand for the conditional probabilities of no mutation,
which are supposed to be much higher than the off-diagonal
probabilities. It is even reasonable to assume that these diagonal
entries are greater than 0.5, giving in particular $a+d+g+f>1$. In
any case, the result proved above can be easily adapted to the
case $a+d+g+f<1$ or $f>g$ (we have not done it here in order to
make the paper more readable). Note also that any $\SSM$ matrix
with determinant $K$ and $f>g$ gives rise to a $\SSM$ matrix with
$f<g$ and determinant $K$ by permuting its 1st and 4th rows and
its 2nd and 3rd rows (or columns, if preferred).

The hypothesis $b\neq c$ was added to simplify the statement of
the Theorem and can be easily removed. Indeed, a matrix $A$ as in
(b) has $b=c$ and determinant equal to $K$ if and only if $F(A)$
has $\beta'=0$ and $K$ is equal to $(\lambda+\mu-1)\al\beta.$
Therefore $A$ is strictly stochastic with determinant $K$ and
$b=c$ if and only if $\frac{K}{\lambda(\lambda+\mu-1)}<|\al|<\mu,$
$\beta=\frac{K}{\al(\lambda+\mu-1)},$ $\beta'=$ and $\al'$ is any
number satisfying $|\al'|<1-\mu.$
\end{rk}

\begin{proof}

(a) Let $A$ be defined from $\lambda,\mu$, $\beta, \dots, \al$ as above.
Then $F(A)$ is equal to
$$
B=\left( \begin{array}{ccccc}
        \lambda   & 1-\lambda&0    &0 \\
            1-\mu & \mu  & 0   & 0 \\
             0  &  0  &\al  & \al' \\
             0 &   0  &\beta'  &  \beta   \end{array} \right).
$$
We prove that $A$ is a stochastic matrix using Lemma
\ref{lem_SSM}.

By hypothesis, $(\lambda,\mu)\in \Theta$ and hence $\lambda$ and $\mu$ lie in
$(0,1).$ Moreover, as $(\al,\beta)\in \Omega_{\lambda,\mu}$, we have $0<\al<\mu,$ $|\beta|<\lambda.$ By assumption (i),
$|\beta'|<1-\lambda$ is also satisfied. It remains to prove that
$|\al'|<1-\mu.$ But this follows from conditions (i) and (ii):
$$|\al'|=\frac{|\al\beta-\frac{K}{\lambda+\mu-1} |}{|\beta'|}<1-\mu.$$


Row sums in $A$ are equal to 1 by definition of $a,b,\dots,h$. Moreover, as $B=F(A)$ is obtained from $A$ by a basis change, we have that $\det A=\det B$ and it coincides with $(\lambda+\mu-1)(\al\beta-\al'\beta').$
Thus, by assumption (ii) we have $\det A=K.$

(b) Lemma \ref{rem_transform} tells us that $F(A)$ has the shape
in the statement of the Proposition, and that $\lambda = a+d$,
$\mu=g+f$, $\al=f-g$, $\al'=h-e$, $\beta=a-d$, and  $\beta'=c-b$. By Lemma
\ref{lem_SSM} we have that $\lambda,\mu$ lie in $(0,1)$,
$[\al|<\lambda$, $|\beta|<\lambda$,  $|\al'|<1-\mu$, $|\beta'|<1-\lambda$. Moreover, as we are assuming
$a+d+g+f>1$, $b\neq c$, and $f>g$, we have $\lambda+\mu>1$, $\beta'\neq
0$, and $0<\al<\mu$.

On the other hand, $\det A=K$ implies
$K=(\lambda+\mu-1)(\al\beta-\al'\beta')$ and therefore condition (ii) holds.

The remaining inequality in (i),
$$\frac{|\al\beta-\frac{K}{\lambda+\mu-1}|}{1-\mu}<|\beta'|,$$
holds because $|\al'|$ satisfies (ii) and
$|\al'|<1-\mu$.

We prove now that $(\al,\beta)$ belongs to $\Omega_{\lambda\mu},$ that is,
%
\begin{equation} \label{lastineq}
|\al\beta-\frac{K}{\lambda+\mu-1}|<(1-\lambda)(1-\mu).
\end{equation}
We have just seen that $|\beta'|$ satisfies condition (i), so
$$|\al\beta-\frac{K}{\lambda+\mu-1}|<|\beta'|(1-\mu)$$ and this last term is
$<(1-\lambda)(1-\mu).$ Therefore \eqref{lastineq} is satisfied.

Finally, as $(\al,\beta)$ is a point in $\Omega_{\lambda,\mu}$,
this set is not empty and $(\lambda,\mu)$ belongs to $\Theta$ by
Proposition \ref{prop_SSM}.
\end{proof}

The previous results and their proofs provide the following algorithm for generating any $\SSM$ matrix
$$A=\left( \begin{array}{ccccc}
        a    &b & c    &d \\
            e   & f  &g   & h \\
            h  & g  &f   & e \\
             d &   c  &b  &  a   \end{array} \right).$$
with $a+d+g+f>1$, $f>g$, and $b\neq c.$

\begin{alg}\label{alg_SSM}
\rm(Generation of $\SSM$ matrices with given determinant.)

\noindent\textit{Input:} $K$ in $(0,1).$

\noindent\textit{Output:} A strictly stochastic $\SSM$ matrix $A$
with determinant $K.$
\begin{itemize}
\item[\texttt{Step 1:}] Compute the unique positive root $\nu_0$ of
$r_K(z)$
following Remark \ref{rmk_SSM}.
\item[\texttt{Step 2:}] Take $s$ randomly in $\left[ \nu_0+1,2\right).$
\item[\texttt{Step 3:}] Take $t$ randomly such that $|t|< \min\left\{2-s,\sqrt{\frac{r_K(s-1)}{s-1}}\right\}.$
\item[\texttt{Step 4:}] Set $\lambda=\frac{s+t}{2}$ and $\mu=\frac{s-t}{2}.$
\item[\texttt{Step 5:}] Take $\al>0$ randomly in $\left(\frac{\frac{K}{\lambda+\mu-1}-(1-\lambda)(1-\mu)}{\lambda},\mu\right).$
\item[\texttt{Step 6:}] Choose $\beta$ randomly such that
$$\max\left\{-\lambda, \frac{\frac{K}{\lambda+\mu-1}-(1-\lambda)(1-\mu)}{\al}\right\} <\beta<
\min\left\{\lambda, \frac{\frac{K}{\lambda+\mu-1}+(1-\lambda)(1-\mu)}{\al}\right\}.$$
\item[\texttt{Step 7:}] Choose $\beta'$ randomly such that
$\frac{|\al\beta-\frac{K}{\lambda+\mu-1}|}{1-\mu}<|\beta'| <1-\lambda.$
\item[\texttt{Step 8:}] Set $\al':=\frac{\al\beta-\frac{K}{\lambda+\mu-1}}{\beta'},$
$a:=(\lambda+\beta)/2$,$b:=(1-\lambda-\beta')/2$,$c:=(1-\lambda+\beta')/2$
$d:=(\lambda-\beta)/2$, $e:=(1-\mu-\al')/2$, $f:=(\mu+\al)/2$,
$g:=(\mu-\al)/2,$  and $h:=(1-\mu+\al')/2.$
\item[\texttt{Final:}] Return
$$A=\left( \begin{array}{ccccc}
        a    &b & c    &d \\
            e   & f  &g   & h \\
            h  & g  &f   & e \\
             d &   c  &b  &  a   \end{array} \right).$$
\end{itemize}
\end{alg}

\begin{rk}
\rm As $\SSM$ matrices include $\Ka$ matrices, using Remark \ref{Rmk_exp} we see that there exist matrices produced by the algorithm above that are not of type $\exp(Q)$.
\end{rk}

\section{Generating $\GMM$ matrices with given determinant}

For $\GMM$ matrices we do not have such a general result as in the
previous sections. We do not know how to  generate \textit{any} strictly stochastic $\GMM$ matrix,
but here we explain a way for generating some of them.


We could obtain a strictly stochastic matrix $\GMM$ matrix with determinant equal to $K$ by exponentiating a rate matrix  (i.e. a matrix with row sums equal to 0 and off-diagonal positive entries) with trace equal to $\log K$ (cf. \cite[Theorem 4.19]{ASCB2005}). However, not all $\GMM$ matrices are of this type (see \cite{Culver} and Remark \ref{Rmk_exp}). We use that the product of two strictly stochastic matrices
is again a strictly stochastic matrix in order to obtain a broader class of $\GMM$ matrices. In fact, we multiply a $\GMM$ matrix of type $\exp(Q)$ with determinant $\delta>K$ by a $\SSM$ matrix of determinant $K/\delta.$ We must admit that we do not know how much larger is this class of matrices. The set $V$ of $\GMM$ matrices with determinant $K$ corresponds to an affine variety of dimension $11.$  There are 11 free parameters for a rate matrix $Q$ with given trace, so the matrices of type $\exp(Q)$ lie on a subset of $V$ of dimension 11. Therefore the set of matrices produced by the algorithm below form a subset of maximum dimension of $V,$ and this subset is larger than the set $\{\exp(Q) | Q \textrm{ rate matrix, } \tr Q= K \}.$

\begin{alg}
\rm(Generation of $\GMM$ matrices with given determinant.)

\noindent\textit{Input:} $K$ in $(0,1).$

\noindent\textit{Output:} A strictly stochastic $\GMM$ matrix $A$
with determinant $K.$
\begin{itemize}
\item[\texttt{Step 1:}] Take a random number $t$ in $(\log K,0).$
\item[\texttt{Step 2:}] Generate a random rate matrix $Q$ with nonzero entries and $\tr Q=t.$
\item[\texttt{Step 3:}] Compute $A_0=\exp(Q).$ 
\item[\texttt{Step 4:}] Following algorithm \ref{alg_SSM}, generate a strictly stochastic $\SSM$ matrix $B$ with determinant equal to $K/e^t.$
\item[\texttt{Final:}] Return $A=BA_0.$
\end{itemize}
\end{alg}


\newcommand{\etalchar}[1]{$^{#1}$}

\end{document}